\documentclass[11pt,a4]{article}
\usepackage{asaf}
 
\title{Towards Lower Bounds for Geometric Spanners in High Dimension} 
\date{}
\author{Robert Krauthgamer%
  \thanks{The Harry Weinrebe Professorial Chair of Computer Science.
    This research was supported by the Israel Science Foundation grant \#1336/23.
  Email: \texttt{robert.krauthgamer@weizmann.ac.il}
  } 
  \qquad
  Nir Petruschka\thanks{
    Email: \texttt{nir.petruschka@weizmann.ac.il}
  } 
\\  Weizmann Institute of Science
}

\begin{document}

\maketitle
\begin{abstract}
We study the stretch--size tradeoff for geometric spanners in high-dimensional $\ell_p$ spaces.
Our main contribution is a simple proof of a lower bound shown by Har-Peled, Indyk, and Sidiropoulos [SODA 2013]: 
Every $2$-hop $t$-spanner of the pointset $\{0,1\}^d$ under $\ell_2$ norm
has at least $(2^d)^{1+\Omega(1/t^2)}$ edges. 
Our proof further extends this result to spanners with Steiner vertices. 
In addition, we establish a connection between bounded-hop spanners and general spanners, as follows. 
If every subset $Y$ of an $n$-point metric has a $t$-spanner with at most $\mu|Y|$ edges, 
then the metric has an $O(t)$-hop $O(t)$-spanner of size $O(n(\mu+\log n))$.
Consequently, hop-restricted spanner lower bounds for a metric 
imply lower bounds without hop restriction for one of its subsets.
\end{abstract}

\section{Introduction}

A spanner for a graph $G$ is a subgraph $H \subseteq G$ whose shortest-path distances approximate those of $G$. This notion dates back to early work on network synchronization~\cite{PU89a,PS89}, and since then, spanners have found numerous applications; see the surveys~\cite{Epp00,Zwick01,ahmed2020graph}.
The main measure of spanner quality is the tradeoff between approximation quality and the number of edges. 
A foundational result states that, for every integer $k\geq1$, every $n$-vertex graph admits a spanner with $O(n^{1+1/k})$ edges that approximates its distances within a factor of $2k-1$~\cite{ADDJS93}.

We study geometric spanners, in which the underlying graph is induced by a finite metric space $(X,d)$. Namely, it is the complete weighted graph on $X$, where each edge $\{u,v\}$ has weight equal to the metric distance $d(u,v)$ between its endpoints.
Specifically, we study finite pointsets $X \subset \mathbb{R}^d$, $|X|=n$, 
where the underlying metric is induced by the $\ell_p$ norm for $1 \leq p \leq \infty$.
Formally, we denote by $\ell_p^d := (\mathbb R^d, \|\cdot\|_p)$
and define a \emph{geometric spanner} of stretch $t$ (abbreviated \emph{$t$-spanner})
of a finite set $X \subseteq \ell_p^d$  
to be a graph $G=(X,E)$ in which every edge $\{x,y\}$ has weight $\|x-y\|_p$ and
\[
    \forall x,y \in X, \qquad  \|x-y\|_p\leq d_G(x,y)\leq t\|x-y\|_p. \label{eq:spanner}
\]

Geometric spanners for finite pointsets in $\ell_p^d$ have been studied extensively, especially in low dimension $d=O(1)$.
In this case, for every $0<\varepsilon\leq1$, every $n$-point subset of
$\ell_p^d$ has a $(1+\varepsilon)$-spanner with $n\varepsilon^{-O(d)}$ edges, and analogous bounds hold more generally in metrics of bounded doubling dimension~\cite{HM06}.
Low-dimensional spanners are relatively well understood, 
and a long line of work has also optimized other, usually secondary, quality measures, 
such as degree, lightness, and the number of hops, 
often simultaneously~\cite{CG06b,CLNS15,ES15,Gottlieb15,BLW19,LS23}.
However, the techniques typically used to construct low-dimensional spanners often result in an exponential dependence on $d$, and therefore perform poorly when $d=\Omega(\log n)$. 
This issue is not merely an artifact of the analysis. 
For example, the classical greedy-spanner algorithm outputs very good low-dimensional spanners~\cite{FS20}, 
whereas in the high-dimensional regime it may output far-from-optimal spanners; 
see \Cref{app:greedy-spanner}.
Thus, constructing spanners in high dimension often requires new algorithmic techniques.

The first result for the high-dimensional setting was obtained in~\cite{HIS13} for Euclidean spaces. 
It showed that, for every $t \geq 1$, one can construct an $O(t)$-spanner with $O(n^{1+1/t^2})$ edges using locality-sensitive hashing (LSH); 
see \Cref{def:LSH} for a formal definition. 
Comparable results were later obtained for other $\ell_p$ metrics using similar techniques~\cite{FN22,KP25,KPS25}.
Such spanners also have concrete algorithmic motivations, e.g., 
they have been deployed to build massive similarity graphs 
for clustering and graph learning~\cite{CareyEtAl22}.
For such applications, a small number of hops can be crucial. 
Formally, an \emph{$h$-hop $t$-spanner} for $(X,d)$ is a geometric $t$-spanner, 
in which for every $u,v \in X$, there is a path with at most $h$ edges and length at most $t\cdot d(u,v)$.
It was shown in~\cite{HIS13} that the stretch--size tradeoff mentioned above,
namely, stretch $O(t)$ using $O(n^{1+1/t^2})$ edges, 
is essentially optimal for $2$-hop Euclidean spanners. 
Nevertheless, for general spanners (i.e., without the hop restriction), 
this tradeoff remains poorly understood,
and even the following basic question, which is our main focus, is open.

\begin{question}\label{que:sparse-spanners}
    Does every $n$-point subset of high-dimensional Euclidean space admit an $O(1)$-spanner with $n^{1+o(1)}$ edges?
\end{question}

For general $\ell_p$ spaces, this question is settled only when $p=\infty$,
in which case the answer is negative: 
the classical Fr\'echet embedding maps every finite metric isometrically into $\ell_\infty$, 
and thus the hardness results from the graph setting apply~\cite{TZ05}.
For $\ell_1$, in contrast, the gap in our knowledge is the most pronounced: 
for fixed stretch greater than $2$, 
there is currently no upper bound that is better than the tradeoff known for general metrics, 
and also no lower bound that is truly superlinear (i.e., by a polynomial factor) 
for spanners without a hop bound.
The only relevant result, implicit in~\cite{HIS13}, 
is that the stretch--size tradeoff known for general metrics 
is optimal also for geometric spanners with a $2$-hop restriction.

\subsection{Our Results}

Our main contribution, presented in \Cref{sec:spanner-lb}, is a simpler proof of the lower bound for $2$-hop spanners from~\cite{HIS13} that also allows arbitrary Steiner vertices. Unlike an ordinary spanner, a \emph{Steiner spanner} for a metric $X$ may include additional vertices outside $X$; the original vertices are then called \emph{terminals}. In the following theorem and throughout, Steiner vertices are allowed to be ``non-metric''; that is, they are not required to lie in $\RR^d$. We state our result in $\ell_1$ for simplicity; comparable results hold for general $\ell_p$ metrics (see \Cref{rem:lb-general-ell_p}). 

\begin{theorem}
\label{thm:intro-two-hop}
Every $2$-hop $t$-spanner of the $d$-dimensional Hamming cube $(\{0,1\}^d,\|\cdot\|_1)$ 
has $(2^d)^{1+\Omega(1/t)}$ edges, 
even if Steiner vertices are allowed.
\end{theorem}

The proof in~\cite{HIS13} uses tools that were previously used in~\cite{MNP06} 
to prove LSH lower bounds. 
Specifically, it combines an isoperimetric inequality for the hypercube from~\cite{MNP06} with a sophisticated double-counting argument to lower-bound the size of $2$-hop spanners for the hypercube.
Also the construction of Euclidean $2$-hop spanners in~\cite{HIS13} uses LSH, 
which altogether indicates a close connection between these two objects.
Our proof makes this connection even more direct: 
we reduce constructing an LSH family for the hypercube 
to constructing a $2$-hop spanner for it.
Our reduction essentially constructs an LSH family 
using random walks on a shifted $2$-hop spanner of the hypercube; 
see \Cref{sec:spanner-lb} for a more detailed overview.

Our second contribution, presented in \Cref{sec:hop-reduction}, 
reduces constructing general spanners to constructing bounded-hop spanners, 
as follows. 

\begin{theorem}
\label{thm:hop-reduction}
Let $(X,d)$ be an $n$-point metric for $n\geq2$, and let $t,\mu\geq1$.
Suppose that every subset $Y\subseteq X$ has a $t$-spanner with at most
$\mu|Y|$ edges. Then $X$ has an $O(t)$-hop $O(t)$-spanner with
$O(n(\mu+\log n))$ edges.
\end{theorem}

The reduction is simple. Using the fact that every subset admits a sparse $t$-spanner, 
we first construct spanners on nets of $X$ at all distance scales. 
Additionally, we connect each point in $X$ to its nearest net point at every scale. 
We can then find a bounded-hop path between each pair of points $x,y$ 
by going through the nearby net points at the scale of their distance $d(x,y)$,
and using a spanner path on this net.
It is easy to see that this path has stretch $O(t)$,
and its hop bound follows from the separation between every two net points.

Although technically simple, this reduction has surprising consequences. 
Its contrapositive says that a sufficiently strong lower bound for bounded-hop spanners of a metric $X$  implies a lower bound without any hop restriction for some subset of $X$. 
Generalizing the lower bound of \Cref{thm:intro-two-hop} to more than $2$ hops seems like a potential route to unconditional lower bounds, but unfortunately, the proof of \Cref{thm:intro-two-hop} does not extend beyond two hops. 
This is a genuine obstacle: already with three hops, the Hamming cube has an $O(1)$-spanner of near-linear size $O((\log d) 2^d)$; see \Cref{app:three-hop-hypercube}.  
Finally, we show that \Cref{thm:hop-reduction} can also be used algorithmically 
to construct spanners in $\ell_p$ spaces for $p>2$. This result is presented in \Cref{sec:l_p-reduction}.

\subsection{Related Work}
Recent theoretical progress has focused on the very-low-stretch regime in high-dimensional Euclidean space. 
In the metric Steiner model, the auxiliary vertices are points in the ambient metric (in this case, Euclidean space), and all edge weights are induced by the metric. 
In the more permissive non-metric model, 
the auxiliary vertices are abstract (need not come from the ambient metric) and their edge weights may be chosen freely, 
provided that distances between the original points are preserved up to the stated stretch. 
Nearly quadratic lower bounds are known for metric Steiner spanners with stretch below $\sqrt{2}$, whereas non-metric Steiner points permit undirected $(1+\varepsilon)$-spanners with $n^{2-\Omega(\varepsilon^3)}$ edges~\cite{AZ23}.

A directed spanner is a weighted directed graph in which, for every ordered pair $x,y$ of original points, the shortest directed path from $x$ to $y$ has length between $\|x-y\|_2$ and $(1+\varepsilon)\|x-y\|_2$. In the non-metric Steiner model, directed spanners with $n^{2-\Omega(\varepsilon^2)}$ edges were constructed in~\cite{AZ23}. This was improved to $\widetilde{O}(n^{2-\Omega(\varepsilon)})$ edges\footnote{Throughout, the notation $\widetilde{O}(f)$ hides $\poly(\log f)$ factors, and $O_\alpha(\cdot)$ hides a factor that depends only on $\alpha$.} in~\cite{JPSWZ26}, together with a nearly matching $\widetilde{\Omega}(n^{2-O(\varepsilon)})$ lower bound. Such spanners can be used to approximate transportation distances: the construction in~\cite{AZ23} gives a truly subquadratic $(1+\varepsilon)$-approximation algorithm for Earth-Mover Distance, and the construction in~\cite{JPSWZ26} gives faster algorithms for general Wasserstein distances.

\section{Preliminaries}

\paragraph{Nets.}
\label{def:net}
Let $(X,d_X)$ be a metric space. A set $N\subseteq X$ is called an
\emph{$r$-net} if it is both $r$-covering and $r$-separated: 
every $x\in X$ has some $y\in N$ with $d_X(x,y)\leq r$, 
and all distinct $u,v\in N$ satisfy $d_X(u,v)>r$. 

\paragraph{The Hamming cube.}
We view $\HH$ as the vector space $\mathbb{F}_2^d$. For $u,x\in\HH$,
we write $u\oplus x$ for coordinate-wise addition modulo $2$. 

\paragraph{Locality-sensitive hashing.}
\label{def:LSH}
We use the standard definition from~\cite{IM98}. 
Let $(X,d_X)$ be a metric space and let $U$ be any set. 
A distribution $\mathcal{H}$ over maps $h:X\to U$ 
is called \emph{$(r,R,p,q)$-sensitive} 
if, for all $x,y\in X$,
\begin{align*}
 d_X(x,y)\leq r
 &\quad\longrightarrow\quad
 \Pr_{h\sim\mathcal{H}}[h(x)=h(y)]\geq p,\\
 d_X(x,y)>R
 &\quad\longrightarrow\quad
 \Pr_{h\sim\mathcal{H}}[h(x)=h(y)]\leq q.
\end{align*}
We will need the following LSH lower bound to prove our main result in \Cref{sec:spanner-lb}.
\begin{theorem}[Proposition 2.1 in~\cite{MNP06}]
\label{thm:mnp}
Let $\mathcal{H}$ be an $(r,R,p,q)$-sensitive family for the metric $(\HH, \|\cdot\|_1)$. 
If $r$ is an odd integer and $R<d/2$, then
\[
 p\leq
 \left(
   q+\exp\left(-\frac{1}{d}\left(\frac{d}{2}-R\right)^2\right)
 \right)^{\frac{e^{2r/d}-1}{e^{2r/d}+1}}.
\]
\end{theorem}

\section{Spanner Lower Bound}\label{sec:spanner-lb}
This section proves \Cref{thm:intro-two-hop}, that is, a lower bound for $2$-hop spanners of the Hamming cube. 
The main novel ingredient is the following reduction from such spanners to LSH families. 
Combining it with the LSH lower bound from~\cite{MNP06}, stated in \Cref{thm:mnp}, yields \Cref{thm:intro-two-hop}.

\begin{lemma}
\label[lemma]{lem:LSH-to-2-hop-spanner-reduction}
Let $G$ be a Steiner $2$-hop $t$-spanner of $(\HH,\|\cdot\|_1)$ with $m$ edges. 
Then, for every $r\geq 1$, there is an $(r,2tr,p,0)$-sensitive family on $(\HH,\|\cdot\|_1)$ 
with $p\geq \Omega\left(\left(\frac{2^d}{m}\right)^2\right)$.
\end{lemma}
We first combine the two components to derive the lower bound for 2-hop spanners of the Hamming cube.

\begin{proof}[Proof of \Cref{thm:intro-two-hop}]
Let $G$ be a $2$-hop Steiner $t$-spanner of $(\HH,\|\cdot\|_1)$ for $t\geq 1$, and denote $m=|E(G)|$.
We may assume that $\frac{d}{t}$ is larger than a sufficiently large constant,\
as otherwise the proof follows from the connectivity of $G$, 
which implies that $m\geq\Omega(|\HH|) = \Omega(2^d) = (2^d)^{1+\Omega(1/t)}$.
Choose an odd integer $r=\Theta(\frac{d}{t})$ such that $2tr\leq \frac{d}{4}$. By \Cref{lem:LSH-to-2-hop-spanner-reduction} and \Cref{thm:mnp}, there is an $(r,2tr,p,0)$-sensitive family satisfying
\begin{align*}
    \Omega\left(\left(\frac{2^d}{m}\right)^2\right)
    \leq p
    &\leq
    \left(e^{-\frac{1}{d}(\frac{d}{2}-2tr)^2}\right)^%
        {\frac{e^{2r/d}-1}{e^{2r/d}+1}} \leq
    \exp\left(
        -\frac{d}{16} \cdot \frac{e^{2r/d}-1}{e^{2r/d}+1}
    \right) \leq
    \exp(-\Omega(r)) = (2^d)^{-\Omega(1/t)},
\end{align*}
where the last inequality uses $\frac{e^{2r/d}-1}{e^{2r/d}+1}=\Theta(\frac{r}{d})$. Rearranging gives $m\geq (2^d)^{1+\Omega(1/t)}$, concluding the proof.
\end{proof}

Before proving \Cref{lem:LSH-to-2-hop-spanner-reduction}, let us explain its main idea, 
namely, how to construct an LSH for $\HH$ from a sparse $2$-hop spanner. 
Consider terminals $x,y\in\HH$ and suppose they are close. 
The random hash function tries to ``guess'' the middle vertex $z$ of a short path $x\to z\to y$ in $G$, that is, each of $x$ and $y$ is hashed to a uniformly random neighbor in $G$. 
If both $x$ and $y$ have low degree, the probability that they both select $z$ is relatively large.
Two major gaps in this simple scheme need to be resolved. 
Far points might collide, 
and some terminals might have high degree even though the average degree in $G$ is small.
To resolve the first issue, we eliminate far collisions by deleting from $G$ long edges. 
A short path $x\to z\to y$ survives this deletion, whereas two far terminals cannot retain a common neighbor.
To handle high degrees, we randomly shuffle $\HH$ by choosing a uniformly random $x\in\HH$ and mapping every terminal $u$ to $u\oplus x$ before hashing. This random shuffling preserves all distances and maps each terminal to one whose degree is at most a constant factor of the average degree with constant probability.
We proceed to give the formal proof. 

\begin{proof}[Proof of \Cref{lem:LSH-to-2-hop-spanner-reduction}]
Denote $n=2^d$. Delete from $G$ every edge of weight greater than $tr$ and add a self-loop at every terminal $u\in\HH$. Let $\Delta$ denote the average degree of the terminals in the resulting graph. Deleting edges cannot increase degrees, while adding the self-loops
increases each terminal degree by at most one. Thus, $\Delta\leq \frac{2m}{n}+1=O(\frac{m}{n})$.

We define $\cH$ by the following random procedure for hashing the cube vertices. First, sample $x\in\HH$ uniformly. Then, independently for every terminal $u\in\HH$, perform a one-step random walk from $u\oplus x$ on $G$ and hash $u$ to the endpoint of the walk. Each outcome of this procedure defines one hash function.

We first show that far points never collide. Let $u,w\in\HH$ such that $\|u-w\|_1>2tr$. If $h(u)=h(w)$, the two sampled edges form a path of weight at most $2tr$ from $u\oplus x$ to $w\oplus x$. Hence
\[
    \|u-w\|_1
    =\|(u\oplus x)-(w\oplus x)\|_1
    \leq d_G(u\oplus x,w\oplus x)
    \leq 2tr,
\]
a contradiction.

We next bound the probability that close points are hashed to the same value.
Let $u,w\in\HH$ be distinct and satisfy $\|u-w\|_1\leq r$.
For every fixed $x$, the original spanner contains a path of at most two edges and total weight at most $tr$ between $u\oplus x$ and $w\oplus x$.
Every edge of this path survives the deletion. Using a self-loop if the original path has one edge, the two shifted endpoints therefore share a common neighbor in $G$.
Each shifted endpoint is uniformly distributed over $\HH$, whose average terminal
degree in $G$ is $\Delta$. By Markov's inequality, each endpoint has degree at most $3\Delta$ with probability at least $2/3$.
Hence, by a union bound, both endpoints have degree at most $3\Delta$ with probability at least $1/3$.
Given this event, the two independent random-walk steps reach their common neighbor with probability at least $1/(3\Delta)^2$.
Therefore,
\[
    \Pr[h(u)=h(w)]\geq \frac{1}{3}\cdot\frac{1}{(3\Delta)^2}
    \geq
    \Omega\left(\left(\frac{n}{m}\right)^2\right),
\]
which concludes the lemma.
\end{proof}

\begin{remark} \label[remark]{rem:lb-general-ell_p}
Since $\|u-v\|_1=\|u-v\|_p^p$ for all $u,v\in\HH$, 
the same proof gives an analogous result in $\ell_p$ with stretch $t^{1/p}$.
For $p=2$, this recovers the Euclidean lower bound from~\cite{HIS13}.
\end{remark}

\section{From Spanners to Hop-Bounded Spanners}
\label{sec:hop-reduction}

This section proves \Cref{thm:hop-reduction}, showing that if every subset admits a sparse spanner, 
then the entire metric admits a sparse hop-bounded spanner.
To present the main idea clearly, 
we first prove the following weaker version of \Cref{thm:hop-reduction}, which has dependence on the aspect ratio of the metric, 
and then remove this dependence using a standard but slightly more involved argument.
Finally, we provide another application of \Cref{thm:hop-reduction}, 
for the construction of spanners in general $\ell_p$ spaces.

\begin{theorem}
\label{thm:hop-reduction-aspect-ratio}
Let $(X,d)$ be an $n$-point metric with aspect ratio $\Phi=\frac{\max_{u\neq v}d(u,v)}{\min_{u\neq v}d(u,v)}$.
Let $t,\mu\geq 1$, and suppose that every subset $Y\subseteq X$ has a $t$-spanner with at most $\mu|Y|$ edges.
Then $X$ has an $O(t)$-hop $O(t)$-spanner with $O\bigl(\mu n(1+\log\Phi)\bigr)$ edges.
\end{theorem}

The idea is simple: If a set of points is $r$-separated, then every path of length $s$ between them contains at most $s/r$ edges. 
We therefore construct multiple spanners for nets of $X$ instead of constructing a spanner directly on $X$. 

\begin{proof}
After rescaling, assume that the minimum distance in $X$ is $1$. Let $N_0=X$, and for every $1\leq i\leq \lceil\log_2\Phi\rceil$, let $N_i$ be a $2^i$-net of $X$. For each $x\in X$, let $\pi_i(x)$ be a closest point of $N_i$, where $\pi_0(x)=x$. Construct $G$ by adding a $t$-spanner $G_i$ of every $N_i$ and the edge $\{x,\pi_i(x)\}$ for every $x$ and $i$. The graph $G$ has $O(\mu n(1+\log\Phi))$ edges.

It remains to bound the stretch and the number of hops. Fix $x,y\in X$, write $D=d(x,y)$, and choose $i$ such that $D\in[2^i,2^{i+1})$. Let $x'=\pi_i(x)$ and $y'=\pi_i(y)$. Since $d(x,x'),d(y,y')\leq2^i$, we have $d(x',y')<4\cdot2^i$. Hence $G_i$ contains an $x' \to y'$ path of length less than $4t2^i$. Since $N_i$ is $2^i$-separated, this path uses fewer than $4t$ hops. Together with $\{x,x'\}$ and $\{y,y'\}$, it gives an $x \to y$ path with fewer than $4t+2$ hops and length at most $4t2^i+2^{i+1}\leq(4t+2)D$, as needed.
\end{proof}

\paragraph{Removing the Aspect-Ratio Dependence.}

The logarithmic dependence on the aspect ratio comes from constructing a spanner on the entire net and connecting every point directly to the net at every level. We remove it using an argument inspired by~\cite{HIS13}. At each level, we construct a spanner only on a suitable subset of the net. These subsets still contain the points needed to obtain short paths, and the total size of all such subsets is $O(n)$. We then use the tree-shortcutting construction from~\cite{AS87} to replace all direct point-to-net connections with $O(n\log n)$ edges, thereby ensuring that net points can be reached via two-hop paths.

\begin{proof}[Proof of \Cref{thm:hop-reduction}]

Rescale so that the minimum distance in $X$ is $1$, and let $\Phi$ be its aspect ratio. We now use hierarchical $2^i$-nets $X=N_0\supseteq N_1\supseteq\cdots\supseteq N_H$,
where $H=\lceil\log_2\Phi\rceil+1$ and $N_H$ is a singleton. Consider an auxiliary net tree which has one node $(i,z)$ for each $z\in N_i$. For $i<H$, the parent of $(i,z)$ is a node $(i+1,z')$, where $z'\in N_{i+1}$ satisfies $d(z,z')\leq2^{i+1}$; when $z\in N_{i+1}$, choose $z'=z$. Give this edge length $d(z,z')$. The root is the unique node at level $H$. Let $\pi_i(x)$ denote the point labeling the level-$i$ ancestor of $(0,x)$.
For convenience, set $N_i=N_H$ for $i>H$.

Call $z\in N_i$ \emph{active} if another point of $N_i$ lies within distance $8\cdot2^i$, and let $A_i \subset N_i$ be the set of active points.

\begin{claim}
\label[claim]{clm:active-count} $\sum_{i\geq0}|A_i|\leq10n$.
\end{claim}

\begin{proof}
For every $i$, partition $A_i$ into sets of size at least $2$, each of diameter at most $2^{i+5}$. This can be done by taking a maximal matching between points in $A_i$ whose distance is at most $8\cdot2^i$ and then assigning each unmatched point to a part containing a point within distance $8\cdot2^i$ from it.
Since $N_{i+5}$ is $2^{i+5}$-separated, it contains at most one point from each part, and therefore $|N_{i+5}|\leq |N_i\setminus A_i|+\frac{|A_i|}{2}$. Rearranging gives
$|A_i| \leq 2(|N_i|-|N_{i+5}|)$, and thus
\[
\sum_{i\geq0}|A_i| \leq \sum_{i\geq0} 2(|N_i|-|N_{i+5}|) \leq 2\sum_{i=0}^4 |N_i| \leq 10n.
\]
\end{proof}

For every $i$, add a $t$-spanner $G_i$ of $A_i$. By \Cref{clm:active-count}, these spanners add $O(\mu n)$ edges in total.

It remains to replace the direct point-to-net edges, which would still cause the spanner size to depend on $\Phi$. To do so, we first prune the auxiliary net tree. Keep the level-zero nodes, the active nodes, the root, and every node with at least two children. Connect every remaining node to its lowest remaining ancestor, or to the root if it has no other remaining ancestor. \Cref{clm:active-count} bounds the number of active occurrences, and a rooted tree with $n$ leaves has fewer than $n$ nodes with at least two children. Thus, the pruned tree has $O(n)$ nodes. Give each edge of the pruned tree the length of the corresponding path in the original net tree. The path from $x$ to a remaining level-$i$ ancestor has length at most $\sum_{j=1}^i2^j$; in particular,
\begin{equation}
\label{eq:nested-net-cover}
    d(x,\pi_i(x))\leq\sum_{j=1}^i2^j<2^{i+1}.
\end{equation}

Now, apply the tree-shortcutting construction from~\cite{AS87}, that is, recursively choose a centroid (a vertex whose removal leaves components of size at most half the tree), connect it to all its ancestors and descendants using their tree distances as edge weights, and continue in every component left after deleting it. The recursion has logarithmic depth and therefore adds $O(n\log n)$ edges. Moreover, when an ancestor-descendant pair is first separated, the chosen centroid lies on its tree path, giving a path of at most two edges with the same length. Replacing each of these edges by the metric edge between its labels can only decrease its length.

Add these edges and the spanners $G_i$. The resulting graph has $O(n(\mu+\log n))$ edges. To verify the stretch and the number of hops, fix distinct $x,y\in X$, write $D=d(x,y)$, and choose $i$ such that $D\in[2^i, 2^ {i+1})$ and put $x'=\pi_i(x)$ and $y'=\pi_i(y)$ where by \eqref{eq:nested-net-cover}, $d(x,x'),d(y,y')<2^{i+1}$.

If $x'=y'$, let $z$ be the lowest common ancestor of $x$ and $y$ in the net tree. It has at least two children and lies at a level at most $i$, so it was retained. Routing through $z$ therefore uses at most $4$ hops and has length less than $2^{i+2} \leq 4D$.
Otherwise, $x'\neq y'$, and hence, by the triangle inequality, $d(x',y') < 2\cdot2^{i+1} + D \leq 6\cdot2^i$. Therefore, both $x'$ and $y'$ are active points, and thus they survived the pruning.
Route from $x$ to $x'$ in at most two hops, follow a shortest path in $G_i$ from $x'$ to $y'$, and then route from $y'$ to $y$ in at most two hops. The middle path has length less than $6t2^i$ and uses fewer than $6t$ hops. Thus, the resulting path has at most $6t+4$ hops and length less than
\[
    2^{i+2}+6t2^i\leq(6t+4)D\leq10tD,
\]
as needed.
\end{proof}

\subsection{Application to $\ell_p$-Spanners}\label{sec:l_p-reduction}

We proceed to give an additional application of \Cref{thm:hop-reduction}; specifically, we show that spanner bounds for $\ell_2$ translate into spanner bounds for $\ell_p$ for every fixed $p>2$, as follows. 

\begin{theorem}
\label{thm:lp-from-l2} 
Let $2<p<\infty$ and $n\geq2$. Suppose that, for some function $s:[1,\infty)\to[1,\infty)$, every finite $Y\subset\ell_2$ admits an $s(t)$-spanner with $O(|Y|^{1+1/t})$ edges for every $t\geq1$. Then, for every $t \geq 1$, every $n$-point $X\subset\ell_p$ of aspect ratio $\operatorname{poly}(n)$ admits an $O_p\bigl(t^{1/2-1/p}s(2t)\bigr)$-spanner with $\Tilde{O}\bigl(n^{1+1/t}\bigr)$ edges.
\end{theorem}

The best currently known spanner construction for $n$-point subsets of $\ell_p$, $p>2$, is given in~\cite{KP25,KPS25}. It yields $O_p(\sqrt{t})$-spanners with $\Tilde{O}(n^{1+1/t})$ edges for $t\geq1$, as in the $\ell_2$ case. Hence, \Cref{thm:lp-from-l2} currently does not give an improved construction in these spaces.
However, unlike \Cref{thm:lp-from-l2}, the construction in~\cite{KPS25} does not follow from a direct reduction to $\ell_2$, where much better results are often available than in the general $\ell_p$ setting.

To prove \Cref{thm:lp-from-l2}, we use two tools that also underlie the construction of $\ell_p$ spanners in~\cite{KPS25}. The first is a probabilistic decomposition for point sets in $\ell_p$.
\begin{definition}
Let $(X,d)$ be a finite metric space, let $\tau\geq1$, $\Delta>0$, and $\eta\in(0,1]$. A distribution $\mathcal D$ over partitions of $X$ is \emph{$(\tau,\Delta,\eta)$-capped} if every cluster has diameter at most $\Delta$ and, for every $x,y\in X$ with $d(x,y)\leq\Delta/\tau$,
\[
 \Pr_{P\sim\mathcal D}[P(x)=P(y)]\geq\eta.
\]
We say that $X$ admits a $(\tau,\eta)$-capped decomposition if it admits a $(\tau,\Delta,\eta)$-capped decomposition for every $\Delta>0$. 
\end{definition}

\begin{theorem}[\protect{\cite[Remark~1.6]{KPS25}}]
\label{thm:lp-capped}

For every $p>2$ and $t\geq1$, every $n$-point subset of $\ell_p$ admits an $(O_p(\sqrt{t}),n^{-1/t})$-capped decomposition.
\end{theorem}
The second tool is the classical embedding commonly known as the Mazur map. It acts coordinatewise, raising the absolute value of each coordinate to the power $p/2$ while preserving its sign. We use a property of a shifted and scaled version of this map; see~\cite[Chapter 9]{benyamini1998geometric}. The formulation below is taken from~\cite[Corollary 2.3]{KP25}.
\begin{theorem}
\label{thm:mazur}
Let $2<p<\infty$.
Every set \(X \subset \ell_{p}\) with diameter at most \( C_0>0\)
admits an embedding \(M_X: X \to \ell_{2}\) such that
\[
  \forall x,y \in X,
  \qquad
  \tfrac{2}{p} (2 C_0)^{1 - p/2} \|x - y\|_{p}^{p/2}
  \leq \|M_X(x) - M_X(y)\|_{2}
  \leq \|x - y\|_{p} .
\]
\end{theorem}

The reduction follows the spanner construction from~\cite{HIS13,FN22}, which reduces the problem to constructing spanners for sets of bounded diameter and then maps each such set to $\ell_2$ using the Mazur map.
For each image, we construct a Euclidean spanner and pull its edges back to $\ell_p$. The guarantees in \Cref{thm:mazur} are nonlinear, so obtaining meaningful bounds in the last step requires a hop bound for the Euclidean spanner; this is where we use \Cref{thm:hop-reduction}.

\begin{proof}[Proof of \Cref{thm:lp-from-l2}]
Rescale $X$ so that its minimum distance is $1$.
For every $1\leq i\leq O(\log n)$, sample $L=O(n^{1/(2t)}\log n)$ independent partitions from an $(O_p(\sqrt{t}),O_p(\sqrt{t})2^i,n^{-1/(2t)})$-capped decomposition given by \Cref{thm:lp-capped}. With high probability, every pair $x,y\in X$ with $\|x-y\|_p\leq2^i$ is clustered together in at least one of these partitions. For every cluster $C$, apply the map $M_C:C\to\ell_2$ given by \Cref{thm:mazur} and compute an $s(2t)$-spanner $G_C$ with $O(|C|^{1+1/(2t)})\leq O(n^{1/(2t)}|C|)$ edges for $M_C(C)$. We then use \Cref{thm:hop-reduction} to convert $G_C$ into an $O(s(2t))$-hop $O(s(2t))$-spanner $\Tilde{G}_C$ of size $\Tilde{O}(|C|^{1+1/(2t)})\leq\Tilde{O}(n^{1/(2t)}|C|)$ and add the corresponding edges of $\Tilde{G}_C$ between the original points of $C$.
For every partition, the cluster sizes sum to $n$, so summing over all levels and sampled partitions gives at most
\[
 O\bigl((\log n)Ln(n^{1/(2t)}+\log n)\bigr)
 =\Tilde{O}\bigl(n^{1+1/t}\bigr)
\]
edges.

It remains to bound the stretch. Fix $x,y\in X$, and let $i$ be the minimal integer such that $\|x-y\|_p\leq2^i$. With high probability, both $x$ and $y$ are contained in a cluster $C$ induced by one of the partitions sampled from the $(O_p(\sqrt{t}),O_p(\sqrt{t})2^i,n^{-1/(2t)})$-capped decomposition. Let $\Delta_i=O_p(\sqrt{t}\,2^i)$ denote the diameter bound of $C$. By the minimality of $i$, $\Delta_i=O_p(\sqrt{t}\,\|x-y\|_p)$.
Let $M_C(v_0)\to \ldots \to M_C(v_l)$ be the $O(s(2t))$-hop path between $M_C(x)$ and $M_C(y)$ in the Euclidean $O(s(2t))$-spanner, where $v_0=x$, $v_l=y$, and $l=O(s(2t))$ and consider the corresponding path $x=v_0 \to \dots \to v_l=y$ in the spanner. By \Cref{thm:mazur} and the Euclidean stretch,

\begin{align*}
 \frac{2}{p}(2\Delta_i)^{1-p/2}
 \sum_{j=1}^{l}\|v_j-v_{j-1}\|_p^{p/2}
 &\leq \sum_{j=1}^{l}
 \|M_C(v_j)-M_C(v_{j-1})\|_2\\
 &\leq O(s(2t))\|M_C(x)-M_C(y)\|_2\\
 &\leq O(s(2t))\|x-y\|_p,
\end{align*}
rearranging gives
\begin{align}\label{eq:mazur-path-bound}
    \sum_{j=1}^{l}\|v_j-v_{j-1}\|_p^{p/2} \leq O_p(\Delta_i^{p/2-1} s(2t) \|x-y\|_p).
\end{align}
Therefore, 
\begin{align*}
 \sum_{j=1}^{l}\|v_j-v_{j-1}\|_p
 &\leq l^{1-2/p}
 \left(\sum_{j=1}^{l}\|v_j-v_{j-1}\|_p^{p/2}\right)^{2/p}
 && \text{by H\"older's inequality,}\\
 &\leq O_p\left(
 l^{1-2/p}\Delta_i^{1-2/p}s(2t)^{2/p}
 \|x-y\|_p^{2/p}
 \right)
 && \text{by \eqref{eq:mazur-path-bound},} \\
 &\leq O_p\left(
 \Delta_i^{1-2/p}s(2t)\|x-y\|_p^{2/p}
 \right)
 && \text{since $l=O(s(2t))$,}\\
 &\leq O_p\left(
 \left(\frac{\Delta_i}{\|x-y\|_p}\right)^{1-2/p}s(2t)
 \right)\|x-y\|_p\\
 &\leq O_p\left(t^{1/2-1/p}s(2t)\right)\|x-y\|_p
 && \text{since $\Delta_i=O_p(\sqrt{t}\,\|x-y\|_p)$,}
\end{align*}
which proves the theorem.
\end{proof}

We remark that the assumption that $X$ has aspect ratio $\operatorname{poly}(n)$ can be removed using hierarchical nets as in~\cite{HIS13}. Additionally, if the $\ell_2$ spanners in \Cref{thm:lp-from-l2} are hop-bounded, their hop bound can be used instead of the hop bound supplied by \Cref{thm:hop-reduction}. In particular, using the $O(1)$-hop $\ell_2$ spanners of~\cite{HIS13} yields bounds similar to those of~\cite{KPS25}; we omit the details.

\paragraph*{Acknowledgments.}
An initial version of the proof of \cref{thm:three-hop-hypercube} was provided by ChatGPT 5.6-sol; the authors subsequently simplified it before its inclusion in the paper. ChatGPT 5.6-sol also assisted with writing and auditing the paper and helped simplify the proof of \cref{thm:hop-reduction} and remove polylogarithmic factors from its size bound. The authors verified the resulting text and proofs and assume responsibility for all content.

{\small
  \bibliographystyle{alphaurl}
  \bibliography{references}
} %

\appendix
\section{A Near-Linear $3$-Hop Spanner for the Hypercube}
\label{app:three-hop-hypercube}

In this appendix, we prove that the full Hamming cube has a near-linear $3$-hop spanner with constant stretch; in particular, the $2$-hop lower bound from \cref{thm:intro-two-hop} no longer holds when one additional hop is allowed.

\begin{theorem}
\label{thm:three-hop-hypercube}
For every $d\geq1$, the $d$-dimensional Hamming cube $(\HH,\|\cdot\|_1)$ has a $3$-hop $128$-spanner with $O\bigl(2^d\log d\bigr)$ edges.
\end{theorem}

The construction uses the classical Hamming code~\cite{Hamming1950}. Specifically, it uses the fact that for every $L=2^\ell-1$, there is a binary Hamming code $H_L\subseteq\{0,1\}^L$ of minimum distance $3$ and size $2^L/(L+1)$.

\begin{proof}
We start by explaining the construction.
For every level $3 \leq \ell \leq O(\log d)$, let $L=2^\ell-1 \geq 7$ and divide the coordinates into $k_L= d/L $ blocks of length $L$. For simplicity, when treating a particular length $L$, assume that $L$ divides $d$; we remove this assumption at the end.
The radius-$1$ balls around the codewords of $H_L$ partition $\{0,1\}^L$ into $m=2^L/(L+1)$ parts. Apply this partition in every block, thereby partitioning $\HH$ into $m^{k_L}$ cells. We call the product of $k_L$ codewords defining each cell the center of the cell, and for every $x \in \HH$, let $\pi_L(x)$ denote the center of the cell containing $x$.
Add every edge $\{x,\pi_L(x)\}$ and connect two distinct centers whenever their distance is at most $3k_L/8$. Finally, connect all points to $0^d$ via a direct edge.

Next, we bound the number of edges.
Fix a center $x$, and consider a neighboring center. Suppose that it differs from $x$ in $r$ blocks, and that the codeword distances in those blocks are the integers $t_1,\ldots,t_r$, where $t_1 + \cdots + t_r \leq 3k_L/8$.
There are at most $2^{k_L}$ ways to choose the $r$ differing blocks, and for a fixed choice of blocks, there are at most $2^{3k_L/8}$ possible vectors $(t_1,\ldots,t_r)$. In each block with codeword distance $t_i \geq 3$, there are at most $\binom L{t_i}\leq L^{t_i}$ choices for the new codeword.
Hence the number of centers realizing these distances is at most $\prod_{i=1}^r L^{t_i}=L^{t_1+\cdots+t_r}\leq L^{3k_L/8}$. Combining the three choices, every center has at most
\[
    2^{k_L+3k_L/8}L^{3k_L/8}
    =\bigl(2^{11/8}L^{3/8}\bigr)^{k_L}
    \leq(L+1)^{k_L}
\]
neighbors, where the last inequality uses $L\geq7$. Since there are only $m^{k_L}=2^d/(L+1)^{k_L}$ centers, there are at most $2^d$ center-center edges for each level. Connecting all points to their centers and to $0^d$ adds $O(2^d)$ edges. Since there are only $O(\log d)$ levels, the spanner has $O(2^d\log d)$ edges.

It remains to bound the stretch. Fix $u,v\in\HH$. If $\|u-v\|_1>d/64$, the path through $0^d$ has at most two edges and length at most $2d<128\|u-v\|_1$. Otherwise, choose the largest length $L=2^\ell-1$, for which $k_L/8\geq\|u-v\|_1$. Such a length exists because $L=7$ gives $k_L/8=d/56\geq d/64$. The next Hamming code length is $2L+1$, so maximality and $2L+1<4L$ give
\[
    k_L<4k_{2L+1}<32\|u-v\|_1.
\]
If $\pi_L(u)=\pi_L(v)$, we immediately get a path of length at most $2k_L<128\|u-v\|_1$. Otherwise, the edge $\{\pi_L(u),\pi_L(v)\}$ was added to the spanner because the cells containing $u$ and $v$ are at distance at most $\|u-v\|_1\leq k_L/8$ and the distance between the centers of two cells is at most three times the distance between the cells. Therefore, $u \longrightarrow \pi_L(u) \longrightarrow \pi_L(v) \longrightarrow v$ is a $3$-hop path in the spanner of length at most
\[
    2k_L+3\|u-v\|_1<128\|u-v\|_1.
\]

Finally, to remove the assumption that $L$ divides $d$, write $d=k_L L+q$, where $k_L=\lfloor d/L\rfloor$ and $0\leq q<L$. We apply the partition only to the first $k_L L$ coordinates. Every $q$-bit suffix is copied unchanged to the center. The number of centers is $m^{k_L}2^q=\frac{2^d}{(L+1)^{k_L}}$, specifically, the factor $2^q$ causes no additional loss in the center count. In the center-center edge count, if two suffixes are at distance $j$, they contribute $\binom qj\leq L^j$ choices, while the number of distance vectors increases from at most $2^{3k_L/8}$ to at most $2^{3k_L/8+1}$. Hence the edge bound changes only by a factor of $2$. Since the suffix is copied unchanged and $k_L$ differs from $d/L$ by less than one, the stretch estimates are unchanged.
\end{proof}

\section{The Greedy Spanner in High Dimension}
\label{app:greedy-spanner}

In this appendix, we show that, in the high-dimensional setting, the classical greedy spanner algorithm~\cite{ADDJS93} may produce spanners that are far from optimal. 
To construct a $t$-spanner $H$ for a weighted graph $G=(V, E, w)$, the greedy spanner algorithm starts with an empty graph, considers every edge of $G$ in nondecreasing order of weights, and adds an edge $\{u,v\}$ whenever the current distance in $H$ between $u$ and $v$ is greater than $t \cdot w(\{u,v\})$. Edges with equal weights may be ordered arbitrarily.

Fix $t>1$ and $n\geq2$, and let $G_{n,t}$ be an $n$-vertex graph with the largest possible number of edges among graphs of girth greater than $\lfloor t\rfloor+1$. For $1\leq p\leq\infty$, consider the standard basis vectors $e_1,\ldots,e_n\in\ell_p^n$, and identify vertex $i$ of $G_{n,t}$ with $e_i$.

\begin{theorem}
\label{thm:greedy-not-optimal}
For every $1\leq p\leq\infty$, there is an ordering for which the greedy $t$-spanner outputs a spanner with at least $|E(G_{n,t})|$ edges.
\end{theorem}

\begin{proof}
All distances between the standard basis vectors are equal, so every ordering of the pairs is valid. Consider an ordering that first examines the pairs corresponding to the edges of $G_{n,t}$ and then examines all remaining pairs.

When an edge $\{u,v\}\in E(G_{n,t})$ is examined, the current graph is a subgraph of $G_{n,t}$. If it already contained a path from $e_u$ to $e_v$ of length at most $t\|e_u-e_v\|_p$, then this path would use at most $\lfloor t\rfloor$ edges. Together with $\{u,v\}$, it would form a cycle in $G_{n,t}$ of length at most $\lfloor t\rfloor+1$, contradicting its girth. 
Thus, the greedy algorithm adds every edge of $G_{n,t}$.
\end{proof}

It is well known that $|E(G_{n,t})|\geq n^{1+\Omega(1/t)}$ for every fixed $t$ and all sufficiently large $n$; see, for example~\cite[Chapter 3]{AS00}. In contrast, the constructions from~\cite{HIS13,FN22, KPS25} give, for every fixed $1<p<\infty$, $O_p(t)$-spanners with $n^{1+o(1/t)}$ edges as $t\to\infty$. Thus, in high dimension, the greedy spanner can be polynomially denser than spanners produced by other algorithms.

\end{document}